\documentclass[letterpaper, 10 pt, conference]{ieeeconf}

\IEEEoverridecommandlockouts
\usepackage{amsmath}
\usepackage{amssymb}
\usepackage{graphicx}
\usepackage{bm}
\usepackage{url}
\usepackage{booktabs}
\newtheorem{proposition}{Proposition}
\newtheorem{theorem}{Theorem}
\newtheorem{corollary}{Corollary}
\newtheorem{remark}{Remark}
\newtheorem{assumption}{Assumption}

\newcommand{\lop}{letter}

\newcommand{\E}{\mathbb{E}}
\newcommand{\tr}{\operatorname{tr}}
\newcommand{\R}{\mathbb{R}}
\newcommand{\simplex}{\Delta^{m-1}}
\newcommand{\w}{\bm{w}}
\newcommand{\cost}{\bm{c}}

\title{\LARGE \bf
Preference-Conditioned Group-Relative Policy Optimization for\\
Terminally Constrained Multi-Objective Control
}

\author{Md Ragib Rownak\textsuperscript{1},
        Sidra Ghayour Bhatti\textsuperscript{1},
        and Qadeer Ahmed\textsuperscript{1}%
\thanks{\textsuperscript{1}Center for Automotive Research,
  The Ohio State University, 930 Kinnear Rd, Columbus, OH 43212, USA.
  E-mail: \texttt{\{rownak.1, bhatti.39, ahmed.358\}@osu.edu}}%
}

\makeatletter
\def\ps@arxivnotice{%
  \def\@oddhead{}\def\@evenhead{}%
  \def\@oddfoot{\hfil\parbox[t]{\textwidth}{\centering\scriptsize
    This work has been submitted to the IEEE for possible publication.
    Copyright may be transferred without notice, after which this
    version may no longer be accessible.}\hfil}%
  \def\@evenfoot{\@oddfoot}}
\makeatother

\begin{document}

\maketitle
\thispagestyle{arxivnotice}
\pagestyle{empty}

\begin{abstract}
Multi-objective optimal control is commonly addressed by
reinforcement learning with fixed scalarization weights, requiring a
separate training run and constraint enforcement for each trade-off.
This \lop{} extends critic-free group-relative policy optimization to
terminally constrained multi-objective control: a single
preference-conditioned policy is trained over random preferences and
returns a feasible approximation of the Pareto front in one run. Three
structural results are established. Group normalization after
scalarization preserves the prescribed preference, whereas
per-objective normalization induces policy-dependent distortion; with
pairwise distinct normalization scales, only single-objective
preferences remain unchanged. A positive floor prevents the
terminal-constraint multiplier from vanishing and renders the penalty
exact when the floor exceeds a geometric multiplier. Deterministic
terminal violation is bounded by expected stochastic violation plus a
term depending on terminal-state curvature and action variance, so
stochastic feasibility alone does not certify deterministic execution.
On a heavy-duty range-extender powertrain with six stage costs and a
terminal state-of-charge band, a representative training run places all 32
preference-grid solutions in band on six cycles inside and outside
the training pool. Averaged over ten runs, the method yields
more in-band solutions under deterministic execution than a
critic-based conditioned baseline with the same terminal-constraint
handling.
\end{abstract}

\section{INTRODUCTION}

Energy management of electrified powertrains is a finite-horizon
optimal control problem with competing stage costs, among them fuel
consumption, battery aging, catalyst aging and tailpipe emissions.
The only state constraint is imposed at the terminal step, a battery
state-of-charge (SOC) band at trip
end~\cite{c_egan2023,c_rownak2026}. Reinforcement learning (RL)
treatments reduce the cost vector to a single weighted sum before
learning~\cite{c_egan2023}: the weights are fixed a priori, each
trade-off point requires one full training run, and the policy carries
no information about neighboring trade-offs. Multi-objective studies
retain this structure: a multi-agent formulation with per-agent
rewards~\cite{c_li2025ems} and a hierarchical formulation with a
weighted multi-term reward~\cite{c_coskun2024} both optimize a single
scalarization and report no Pareto front. A Pareto front
from a single training run has not been reported for powertrain energy
management.

In multi-objective RL (MORL), single-policy methods condition the
actor on a preference vector $\w \in \simplex$ and train over
preference draws~\cite{c_lu2023,c_terekhov2024,c_ambadkar2026},
which loses no Pareto points for stationary stochastic
policies~\cite{c_lu2023,c_hayes2022}. These methods are
critic-based: a vector-valued value network is conditioned on $\w$
and per-objective advantages are scalarized inside the surrogate
loss~\cite{c_terekhov2024,c_ambadkar2026}. Contemporaneous
critic-free group-relative work targets language
models~\cite{c_nourzad2026,c_li2026pareto}.
Constrained MORL is less developed: the closest study trains a
preference-conditioned
TD3 policy for airship trajectory planning under two objectives and
expectation constraints~\cite{c_meng2026}, with dual convergence
resting on a dense path-safety cost, multipliers bounded only from
above, and no formal results. None of these studies treats a terminal
equality-band constraint, and none is critic-free.

Group-relative policy optimization (GRPO)~\cite{c_shao2024} replaces
the learned critic of proximal methods~\cite{c_schulman2017} with
group rollout statistics. Its terminally constrained variant,
Advantage-Ranked GRPO (A-GRPO)~\cite{c_rownak2026cdc}, is developed
in the authors' prior work with boundedness and feasibility
properties. With several objectives, existing multi-component
reward methods recommend per-component
normalization~\cite{c_girgis2026,c_ichihara2025}. This \lop{} establishes the
converse for terminal constraints: per-objective normalization
optimizes the policy-dependent distorted preference
$\tilde{w}_i \propto w_i / s_i(\pi)$ for group scales $s_i(\pi)$,
whereas normalization
after scalarization preserves the prescribed preference direction. A
second issue is the terminal multiplier: sustained batch feasibility
drives it geometrically toward zero~\cite{c_rownak2026cdc}, while the
stochastic feasibility test is passed and the deterministic policy
leaves the terminal band. To the authors' knowledge, no constrained RL method bounds the
multiplier from
below~\cite{c_altman1999,c_achiam2017,c_paternain2019,c_stooke2020};
the floor introduced here excludes the decay mode and renders the
penalty exact under a geometric-multiplier condition.

The contributions are:
\begin{enumerate}
\item preference-conditioned A-GRPO, a critic-free method for
  terminally constrained multi-objective control whose objective
  channel reduces to the fixed-weight form at every $\w$
  (Sections~II--III);
\item a preference-distortion analysis characterizing the fixed
  points of per-objective group normalization (Proposition~1,
  Corollary~1); a multiplier floor that rules out vanishing
  constraint pressure, with an exact penalty above a geometric
  multiplier (Theorem~1); a bound separating deterministic from stochastic
  feasibility (Theorem~2); and a weak-Pareto characterization of the
  front (Remark~1) (Section~IV);
\item validation on heavy-duty powertrain energy management with six
  stage costs and a terminal SOC band, where one training covers the
  simplex with band-feasible solutions on cycles inside and outside
  the training pool, at the cost of one fixed-weight run (Section~V).
\end{enumerate}

\section{PROBLEM FORMULATION}

Episodes have a fixed horizon of $T$ steps. At step
$t \in \{0,\ldots,T{-}1\}$, the policy observes
$o_t = h(x_t, d_t) \in \R^{d_o}$ and selects
$a_t \sim \pi_\theta(\cdot \mid o_t)$ from a mixed
continuous and discrete action space. The dynamics
$x_{t+1} = f(x_t, a_t, d_t)$ are deterministic with exogenous input
$d_t$, and $\gamma = 1$. Unlike the single-objective setting of~\cite{c_rownak2026cdc}, the
stage cost is a vector
$\cost_t = (c_{1,t},\ldots,c_{m,t}) \in [0,1]^m$, each component
normalized by a fixed plant reference constant. Here $m = 6$: fuel
mass including an engine-start equivalent, battery aging,
aftertreatment catalyst aging, tailpipe $\mathrm{NO}_x$ mass,
tailpipe $\mathrm{NH}_3$ slip above a soft reference, and a
catalyst-coverage tracking cost. The episode cost of
policy $\pi$ is
\begin{equation}\label{eq:vec_cost}
  \bm{J}(\pi) = \E_{\pi}\Big[\textstyle\sum_{t=0}^{T-1} \cost_t\Big]
  \in \R^m ,
\end{equation}
with components $C_{\mathrm{fuel}}$, $C_{\mathrm{bat,age}}$,
$C_{\mathrm{cat,age}}$, $C_{\mathrm{NO}_x}$, $C_{\mathrm{NH}_3}$ and
$C_{\theta}$ in the order above. The terminal constraint is retained from~\cite{c_rownak2026cdc}: with
$\xi_{\mathrm{term}} = \max(|x_T^{\ast} - x_{\mathrm{ref}}| - \varepsilon,\, 0)$
for the constrained scalar component $x_T^{\ast}$ of the terminal
state, target $x_{\mathrm{ref}}$, and tolerance $\varepsilon > 0$,
the feasible policy class is
\begin{equation}\label{eq:feas_class}
  \Pi_{\varepsilon} = \big\{\pi :
  \Pr\nolimits_{\pi}(\xi_{\mathrm{term}} = 0) \geq p_{\mathrm{safe}}\big\},
\end{equation}
with target feasibility rate $p_{\mathrm{safe}} \in (0,1]$.

For a preference vector $\w \in \simplex = \{\w \in \R^m_{\geq 0}
: \sum_i w_i = 1\}$, the scalarized constrained problem is
\begin{equation}\label{eq:pw}
  P(\w):\quad \min_{\pi \in \Pi_{\varepsilon}}\;
  \w^{\!\top} \bm{J}(\pi).
\end{equation}
A policy $\pi \in \Pi_{\varepsilon}$ is weakly Pareto optimal if no
$\pi' \in \Pi_{\varepsilon}$ satisfies
$\bm{J}(\pi') < \bm{J}(\pi)$ componentwise~\cite{c_miettinen1999}.
The objective is a single conditioned policy
$\pi_\theta(a \mid o, \w)$ whose fixed-preference slices solve
$P(\w)$ across $\simplex$. Auxiliary shaping penalties, for actuator
switching and the terminal band, are not components of $\cost_t$.

\section{PREFERENCE-CONDITIONED A-GRPO}

\subsection{Preference Sampling and Conditioning}

At each update, rollouts are collected in groups of $B$ parallel
trajectories sharing one exogenous input sequence and one preference
vector, drawn once per group as
\begin{equation}\label{eq:wdraw}
  \w \sim (1 - p_{\mathrm{eq}})\,\mathrm{Dir}(\alpha \bm{1}_m)
  + p_{\mathrm{eq}}\,\delta_{\bm{1}_m/m},
\end{equation}
a Dirichlet draw, replaced with probability $p_{\mathrm{eq}}$ by the
equal-preference vector. The draw is shared because the group
statistics that replace the critic must compare rollouts under the
same scalarization. The drawn $\w$ is appended to the policy
observation, $\pi_\theta(a \mid o, \w)$; no critic or other
conditioned network is used.

\subsection{Scalarized Group-Relative Advantage}

The shaped per-step reward is
\begin{equation}\label{eq:rw}
  \tilde{r}^{\w}_t = -\, \kappa_r\, m\, \w^{\!\top} \cost_t
    + r^{\mathrm{con}}_t ,
\end{equation}
where $\kappa_r > 0$ is the common reward scale and
$r^{\mathrm{con}}_t$ collects the terminal-constraint and
tail-shaping Lagrangian terms and the shared shaping penalties,
outside the preference weighting. The factor $m$ preserves the
unit-weight scale, so at $\w = \bm{1}_m/m$ the reward \eqref{eq:rw}
coincides with the fixed-weight composite, and the objective channel
reduces to its scalarized form at every $\w$. For rollout $b$,
$G^{\w}_{t,b} = \sum_{k=t}^{T-1} \tilde{r}^{\w}_{k,b}$ is normalized
over the group at each timestep,
\begin{equation}\label{eq:adv_sc}
  \hat{A}^{\mathrm{sc}}_{t,b}
  = \frac{G^{\w}_{t,b} - \mu_t}{\max(\sigma_t,\, \varphi_t)}\,,
\end{equation}
with $\mu_t$, $\sigma_t$ the group mean and standard deviation of
$G^{\w}_{t,b}$ and $\varphi_t$ the scale-adaptive floor
of~\cite{c_rownak2026cdc}, whose terminal feasibility ranking terms
and clipped surrogate update are applied to \eqref{eq:adv_sc}.

\subsection{Normalization Placement}

Normalization in \eqref{eq:adv_sc} acts after the preference
weighting; the alternative places it before. With per-objective
returns $G_{i,t,b} = -\kappa_r\sum_{k=t}^{T-1} c_{i,k,b}$ and group
statistics $\mu_{i,t}$, $\sigma_{i,t}$,
\begin{equation}\label{eq:adv_po}
  \hat{A}^{\mathrm{po}}_{t,b}
  = \sum_{i=1}^{m} w_i\,
    \frac{G_{i,t,b} - \mu_{i,t}}{\max(\sigma_{i,t},\, \varphi_{i,t})}\,,
\end{equation}
is re-standardized to unit scale over the group at each timestep, and
the constraint terms form a separate channel whose normalizer does not
contain the multiplier. This is the placement recommended for
multi-component rewards in~\cite{c_girgis2026,c_ichihara2025}; the
method uses \eqref{eq:adv_sc}.

\subsection{Floored Dual Update}

Each exogenous input class $c$ carries its own terminal multiplier
$\lambda_c$ on a prescribed interval
$0 < \underline{\lambda}_c \leq \lambda_c^{(0)} < \overline{\lambda}_c$.
With batch mean violation $\bar{\xi}_c$, ascent factor
$\kappa_\lambda > 1$, and decay $\delta_\lambda \in (0,1)$, the
update is
\begin{equation}\label{eq:dual_floor}
  \lambda_c^{(n+1)} = \Pi_{[\underline{\lambda}_c,\,
    \overline{\lambda}_c]}\!\left(
  \begin{cases}
    \kappa_\lambda\, \lambda_c^{(n)}, & \bar{\xi}_c^{(n)} > 0,\\[2pt]
    (1 - \delta_\lambda)\, \lambda_c^{(n)}, & \text{otherwise},
  \end{cases}\right)
\end{equation}
where $\Pi$ denotes projection onto
$[\underline{\lambda}_c, \overline{\lambda}_c]$. The update departs
in two respects from that of~\cite{c_rownak2026cdc}, where ascent is
additive in the violation and the lower bound is zero. Ascent here is
multiplicative, so the multiplier responds to any residual violation
of a preference-randomized group. The floor $\underline{\lambda}_c$
addresses the opposite mode: under sustained batch feasibility the
multiplier of~\cite{c_rownak2026cdc} decays geometrically toward zero
while the deterministic policy leaves the terminal band.

\subsection{Front Extraction}

After training, the policy is evaluated deterministically on a fixed
grid $\mathcal{W} \subset \simplex$ of vertices, the equal-preference
anchor, pairwise midpoints and Dirichlet samples.
The feasible members of
$\{\bm{J}(\pi_\theta(\cdot \mid \cdot, \w)) : \w \in
\mathcal{W}\}$ are filtered by $m$-dimensional dominance, and the
resulting set approximates the Pareto front of \eqref{eq:vec_cost}
over $\Pi_{\varepsilon}$.

\section{PREFERENCE FIDELITY AND CONSTRAINT PRESSURE}

Two choices in Section~III depart from established practice: group
normalization is placed after the preference weighting, and the
terminal multiplier carries a strictly positive floor. The first departs
from the per-component normalization recommended for
multi-component rewards, and Section~\ref{sec:fidelity} shows that
the recommendation does not transfer to a terminally constrained
problem. The results of this section concern the pre-clip policy
gradient
built from \eqref{eq:adv_sc} or \eqref{eq:adv_po} and the dual
update \eqref{eq:dual_floor}. Throughout, one group of $B$ rollouts
sharing $(d, \w)$ is considered at a fixed timestep $t$, and
$s_t = \max(\sigma_t, \varphi_t)$,
$s_{i,t} = \max(\sigma_{i,t}, \varphi_{i,t})$ denote the
normalization scales of \eqref{eq:adv_sc} and \eqref{eq:adv_po}.

\subsection{Preference Fidelity of Normalization Placement}
\label{sec:fidelity}

Write $G^{\mathrm{con}}_{t,b}$ for the return-to-go of
$r^{\mathrm{con}}_t$ alone, so that
$G^{\w}_{t,b} = m \sum_{i} w_i G_{i,t,b} + G^{\mathrm{con}}_{t,b}$,
and let $\mu_{i,t}$, $\mu^{\mathrm{con}}_t$ be the corresponding
group means.

\begin{proposition}[Preference fidelity]\label{prop:fidelity}
For every rollout $b$ of the group:

(i) the scalarized advantage \eqref{eq:adv_sc} satisfies
\begin{equation}\label{eq:sc_identity}
  \hat{A}^{\mathrm{sc}}_{t,b} = \frac{1}{s_t}\Big(
  m \sum_{i=1}^{m} w_i \big(G_{i,t,b} - \mu_{i,t}\big)
  + G^{\mathrm{con}}_{t,b} - \mu^{\mathrm{con}}_t \Big),
\end{equation}
a positive multiple of the group-centered return of the intended
scalarization: the relative weights of the $m$ objectives equal
$\w$, and the constraint terms retain their coefficient relative to
the objectives;

(ii) the per-objective advantage \eqref{eq:adv_po} satisfies
\begin{equation}\label{eq:po_identity}
  \hat{A}^{\mathrm{po}}_{t,b}
  = \Big(\sum_{j=1}^{m}\frac{w_j}{s_{j,t}}\Big)
    \sum_{i=1}^{m} \tilde{w}_{i,t}\big(G_{i,t,b} - \mu_{i,t}\big),
\end{equation}
with $\tilde{w}_{i,t} = (w_i/s_{i,t})/\sum_{j} w_j/s_{j,t}$,
a positive multiple of the group-centered return of the distorted
preference $\tilde{\w}_t \in \simplex$. Unless the $s_{i,t}$ agree
on the active components of $\w$, one has $\tilde{\w}_t \neq \w$,
and $\tilde{\w}_t$ varies with the policy and the timestep through
the group spreads.
\end{proposition}

\textbf{Proof.}
Group means are linear, so
$\mu_t = m\sum_i w_i \mu_{i,t} + \mu^{\mathrm{con}}_t$;
substituting into \eqref{eq:adv_sc} gives \eqref{eq:sc_identity}
with $s_t > 0$. For \eqref{eq:po_identity}, each term of
\eqref{eq:adv_po} is $w_i (G_{i,t,b} - \mu_{i,t})/s_{i,t}$;
factoring the positive constant $\sum_j w_j/s_{j,t}$ and normalizing
the coefficients onto $\simplex$ yields \eqref{eq:po_identity}.
The final re-standardization of \eqref{eq:adv_po} is a further
positive per-timestep scaling and leaves the direction of
\eqref{eq:po_identity} unchanged. \hfill$\square$

Proposition~\ref{prop:fidelity} is the counterpart of the reweighting
established in~\cite{c_girgis2026}, where the scalarize-then-normalize
advantage weights component $j$ by $\lambda_j\sigma_j/\sigma_{RS}$
rather than by the prescribed $\lambda_j$, and per-component
normalization is proposed as the remedy. Each placement preserves one
object and distorts the other. Equation~\eqref{eq:adv_sc} preserves the
objective-space preference $\w$, the quantity that parameterizes the
front over costs already on a common scale.
Equation~\eqref{eq:adv_po} preserves standardized per-component
influence but replaces $\w$ by the policy-dependent $\tilde{\w}_t$.
Because it also changes the objective-to-constraint scale, objectives
with large group spread provide a weaker opposing cost gradient for
the multiplier to overcome.

\begin{corollary}[Fixed points of the distortion]\label{cor:exclusion}
Fix $t$, let $\mathcal{S} = \{i : w_i > 0\}$ and let
$\kappa_t = \max_{i \in \mathcal{S}} s_{i,t} \,/\,
\min_{i \in \mathcal{S}} s_{i,t} \geq 1$. Under \eqref{eq:adv_po}:

(i) $w_i / \kappa_t \leq \tilde{w}_{i,t} \leq \kappa_t\, w_i$ for
every $i$;

(ii) $\tilde{\w}_t = \w$ if and only if $s_{i,t}$ takes one common
value on $\mathcal{S}$.

Hence, if the scales $s_{i,t}$ are pairwise distinct, $\tilde{\w}_t = \w$
holds exactly at the $m$ vertices of $\simplex$: every preference with
more than one active objective, the equal-preference anchor included,
is mapped to a different point of the simplex.
\end{corollary}

\textbf{Proof.}
Write $S = \sum_{j} w_j / s_{j,t} > 0$, so
$\tilde{w}_{i,t} = (w_i/s_{i,t})/S$. (i) For $i \notin \mathcal{S}$
both bounds are $0$. Let $s_- = \min_{j \in \mathcal{S}} s_{j,t}$ and
$s_+ = \max_{j \in \mathcal{S}} s_{j,t}$, so $\kappa_t = s_+/s_-$.
Since $\sum_{j \in \mathcal{S}} w_j = 1$, bounding each term of $S$
gives $1/s_+ \leq S \leq 1/s_-$, and for $i \in \mathcal{S}$,
$w_i/s_+ \leq w_i/s_{i,t} \leq w_i/s_-$; dividing the two chains
yields $w_i/\kappa_t \leq \tilde{w}_{i,t} \leq \kappa_t w_i$. (ii) $\tilde{w}_{i,t} = w_i$ for all $i$ is
equivalent to $w_i(1/s_{i,t} - S) = 0$ for all $i$, that is
$s_{i,t} = 1/S$ on $\mathcal{S}$; conversely, if $s_{i,t} = s$ on
$\mathcal{S}$ then $S = 1/s$ and $\tilde{w}_{i,t} = w_i$.
\hfill$\square$

Under distinct scales, \eqref{eq:adv_po} optimizes the intended
preference only at the vertices; interior displacements are bounded
by the scale ratio, a
property of the current policy that is neither observed nor
controlled when $\w$ is drawn.

\subsection{Persistent Constraint Pressure}

Sustained batch feasibility drives the unfloored multiplier
of~\cite[Thm.~2]{c_rownak2026cdc} geometrically to zero, removing
the constraint gradient exactly when it is still needed; the floor
in \eqref{eq:dual_floor} excludes this mode.

\begin{theorem}[Floored dual update]\label{thm:floor}
For a fixed input class $c$, consider \eqref{eq:dual_floor} with
$0 < \underline{\lambda}_c \leq \lambda_c^{(0)} < \overline{\lambda}_c$.

(i) (Invariance) $\lambda_c^{(n)} \in
[\underline{\lambda}_c, \overline{\lambda}_c]$ for all $n \geq 0$.

(ii) (Persistence) If $\bar{\xi}_c^{(n)} = 0$ for
all $n \geq n_0$, then
$\lambda_c^{(n)} = \max\big(\underline{\lambda}_c,
(1-\delta_\lambda)^{\,n-n_0}\lambda_c^{(n_0)}\big)$, so
$\lambda_c$ reaches $\underline{\lambda}_c$ within
$\lceil \ln(\lambda_c^{(n_0)}/\underline{\lambda}_c) /
\ln\!\big(1/(1-\delta_\lambda)\big) \rceil$ updates and remains
there. The coefficient of the terminal violation in the shaped
reward \eqref{eq:rw} therefore never falls below
$\underline{\lambda}_c$.

(iii) (Exact penalty) Fix $\w$ and suppose the problem
$\min_{\pi} \w^{\!\top}\bm{J}(\pi)$ subject to
$\E_\pi[\xi_{\mathrm{term}}] = 0$ admits a geometric
multiplier $\lambda^{\ast}_{\w}$~\cite{c_bertsekas1999}. If
$\underline{\lambda}_c > \lambda^{\ast}_{\w}$, then for every
$\lambda \geq \underline{\lambda}_c$, each global minimizer of
$\w^{\!\top}\bm{J}(\pi) + \lambda\,\E_\pi[\xi_{\mathrm{term}}]$ is
feasible and solves the constrained problem.
\end{theorem}

\textbf{Proof.}
(i) $\lambda_c^{(0)}$ lies in the interval by hypothesis and the
projection in \eqref{eq:dual_floor} enforces it at every update.
(ii) By (i), $\lambda_c^{(n_0)}
\geq \underline{\lambda}_c$, and under the stated condition only
the decay branch
is active; iterating $\lambda \mapsto (1-\delta_\lambda)\lambda$
and applying the projection gives the expression, and the hitting
time follows by solving
$(1-\delta_\lambda)^{k}\lambda_c^{(n_0)} \leq
\underline{\lambda}_c$ for the smallest integer $k$.
(iii) Let $\pi_\lambda$ minimize the penalized cost and $\pi^{\ast}$
solve the constrained problem. The geometric multiplier property
gives
$\w^{\!\top}\bm{J}(\pi^{\ast}) \leq
\w^{\!\top}\bm{J}(\pi_\lambda) +
\lambda^{\ast}_{\w}\E_{\pi_\lambda}[\xi_{\mathrm{term}}]$,
while optimality of $\pi_\lambda$ gives
$\w^{\!\top}\bm{J}(\pi_\lambda) +
\lambda\,\E_{\pi_\lambda}[\xi_{\mathrm{term}}] \leq
\w^{\!\top}\bm{J}(\pi^{\ast})$. Subtracting,
$(\lambda - \lambda^{\ast}_{\w})\,
\E_{\pi_\lambda}[\xi_{\mathrm{term}}] \leq 0$, and since
$\lambda > \lambda^{\ast}_{\w}$ and
$\xi_{\mathrm{term}} \geq 0$,
$\E_{\pi_\lambda}[\xi_{\mathrm{term}}] = 0$. Feasibility and the
second inequality then give optimality. \hfill$\square$

No convergence claim is made for the coupled policy--multiplier
iteration; the floor trades a bounded conservatism for the exclusion
of the vanishing-pressure mode.

\subsection{Deterministic Execution of a Stochastic Constraint}

Training enforces the terminal constraint through the batch mean
violation $\bar{\xi}_c$ of \eqref{eq:dual_floor}, evaluated on
sampled rollouts, whereas the reported policy is executed
deterministically. The two quantities are related by the curvature of
the terminal-state map.

\begin{assumption}\label{as:soc}
For fixed exogenous input and engine on/off sequences, the
constrained terminal component is $x^{\ast}_T = \Phi(u)$ for a twice
continuously differentiable $\Phi$ of the continuous latent action
sequence $u = (u_0,\ldots,u_{T-1})$, with
$\|\nabla^2 \Phi(u)\| \leq C_\Phi$ on the reachable set. The
stochastic policy draws $u = \bar{u} + \eta$ with $\eta$ of
zero mean, independent across steps, and with per-step covariance
$\Sigma_t$, and the deterministic policy executes $\bar{u}$.
\end{assumption}

\begin{theorem}[Execution gap]\label{thm:gap}
Let $\xi_{\det} = \max(|\Phi(\bar{u}) - x_{\mathrm{ref}}| -
\varepsilon,\, 0)$ be the deterministic terminal violation and
$\bar{\xi} = \E_\eta[\max(|\Phi(\bar{u}+\eta) -
x_{\mathrm{ref}}| - \varepsilon,\, 0)]$ the stochastic violation of
the same policy. Under Assumption~\ref{as:soc},
\begin{equation}\label{eq:gap}
  \xi_{\det} \;\leq\; \bar{\xi}
  \;+\; \tfrac{1}{2}\, C_\Phi \sum_{t=0}^{T-1} \tr \Sigma_t .
\end{equation}
\end{theorem}

\textbf{Proof.}
The map $\xi(x) = \max(|x - x_{\mathrm{ref}}| - \varepsilon, 0)$ is
the distance from $x$ to the interval
$[x_{\mathrm{ref}} - \varepsilon, x_{\mathrm{ref}} + \varepsilon]$,
hence convex and $1$-Lipschitz. Jensen's inequality gives
$\xi(\E[\Phi(\bar{u}+\eta)]) \leq
\E[\xi(\Phi(\bar{u}+\eta))] = \bar{\xi}$. A second-order
Taylor expansion with Lagrange remainder gives
$|\Phi(\bar{u}+\eta) - \Phi(\bar{u}) -
\nabla\Phi(\bar{u})^{\!\top}\eta| \leq
\tfrac{1}{2} C_\Phi \|\eta\|^2$; taking expectations and using
$\E[\eta] = 0$,
$|\E[\Phi(\bar{u}+\eta)] - \Phi(\bar{u})| \leq
\tfrac{1}{2} C_\Phi \E\|\eta\|^2 =
\tfrac{1}{2} C_\Phi \sum_t \tr\Sigma_t$. Since $\xi$ is
$1$-Lipschitz,
$\xi_{\det} \leq \xi(\E[\Phi(\bar{u}+\eta)]) +
|\Phi(\bar{u}) - \E[\Phi(\bar{u}+\eta)]|$, and the two
bounds give \eqref{eq:gap}. \hfill$\square$

If $\Phi$ is affine then $C_\Phi = 0$ and deterministic execution is
at least as feasible as the stochastic average; otherwise the gap is
curvature-driven. Since the log standard deviations are clipped away
from zero, the second term of \eqref{eq:gap} does not vanish, and a
satisfied stochastic test certifies the sampled rollouts and not the
deterministic execution, a statement about the certificate rather
than the size of the realized violation. The floor of
Theorem~\ref{thm:floor} acts on $\bar{\xi}$; dense terminal-band
shaping in $r^{\mathrm{con}}_t$ reduces the remaining curvature
contribution.

\subsection{Pareto Characterization of the Extracted Front}

\begin{remark}[Weak Pareto optimality]\label{rem:pareto}
Fix $\w \in \simplex$ and let $\pi^{\ast}$ solve $P(\w)$ in
\eqref{eq:pw}. No $\pi' \in \Pi_{\varepsilon}$ satisfies
$\bm{J}(\pi') < \bm{J}(\pi^{\ast})$ componentwise, since that
would give
$\w^{\!\top}\bm{J}(\pi') < \w^{\!\top}\bm{J}(\pi^{\ast})$;
hence $\pi^{\ast}$ is weakly Pareto optimal in $\Pi_{\varepsilon}$,
and Pareto optimal when $\w$ is positive~\cite{c_miettinen1999}. The
extraction of Section~III therefore targets weakly Pareto optimal
policies over the feasible class, with only the weak guarantee at the
vertices of $\mathcal{W}$. Linear scalarization loses no Pareto
points for stationary stochastic policies~\cite{c_lu2023}; the
unconstrained conditioned-policy correspondence is established
in~\cite{c_kubo2026}.
\end{remark}

\section{VALIDATION ON HEAVY-DUTY POWERTRAIN ENERGY AND EMISSIONS MANAGEMENT}

\subsection{Setup and Protocol}

The plant is a heavy-duty series-hybrid (range-extender) powertrain
with engine-out emissions and selective catalytic reduction
aftertreatment~\cite{c_rownak2026,c_rownak2026cdc}. The actions are
binary engine on/off, engine power, and the
ammonia-to-$\mathrm{NO}_x$ dosing ratio; the six stage costs are
those of Section~II and the terminal constraint requires
$|x^{\ast}_T - 0.55| \leq 0.002$. Four drive cycles form the training
pool, each with its own multiplier in \eqref{eq:dual_floor}: a
heavy-duty transient cycle (HHDDT), a repeated highway cycle (HWFET)
and two windows of recorded route operation (FDDR-1, FDDR-2). Two further
cycles lie outside the training pool: PDMHC, used for monitoring, and
a third recorded window (FDDR-3). Episode lengths range from
$2{,}947$ to $8{,}156$ steps. The actor is a multilayer perceptron
trained with $B = 48$ rollouts per group for $600$ policy updates
with analytic terminal-credit decomposition. In \eqref{eq:wdraw},
$\alpha = 1$ and $p_{\mathrm{eq}} = 0.25$; in \eqref{eq:dual_floor},
$\kappa_\lambda = 1.03$, $\delta_\lambda = 0.01$, the caps are $3800$
for HHDDT and FDDR-1, $6500$ for HWFET and $10000$ for FDDR-2, and the
floor and initial multiplier are $0.3\,\overline{\lambda}_c$ and
$\sqrt{0.3}\,\overline{\lambda}_c$. The evaluation grid
$\mathcal{W}$ has $32$ preferences.

Ten training runs are performed with seeds fixed in advance; the
preference stream \eqref{eq:wdraw} is shared across runs, so the
extracted fronts are comparable point for point. All reported
evaluations are deterministic and every quantity is reported at the
last training iterate, without checkpoint selection, with front and
SOC detail shown for one representative run. The fixed-weight comparison uses seven
same-backbone trainings, at the equal-preference vector and the six
one-hot vertices, each solving \eqref{eq:pw} at one fixed $\w$ for
the same wall-clock without preference sampling or floor, and
evaluated at its own preference. The second baseline is a preference-conditioned
PPO-Lagrangian (PPO-L) agent~\cite{c_schulman2017,c_stooke2020} that replaces
the critic-free advantage of Section~III by generalized advantage
estimation from a learned $\w$-conditioned value function. All shared
elements match the proposed method, and all shared hyperparameters are
read from the same configuration source. Neither is tuned further, and
both use the same seeds and grid.

\begin{figure}[t]
\centering
\includegraphics[width=\columnwidth]{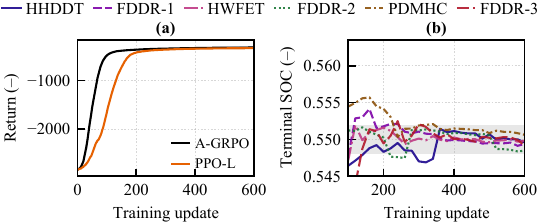}
\caption{(a) Training return: mean over the ten runs with the
25th--75th percentile band. (b) Deterministic terminal SOC at the
equal-preference anchor for the representative run, evaluated every
20 updates on all six cycles; the shaded band is the terminal
constraint $55.0 \pm 0.2\%$.}
\label{fig:training}
\end{figure}

\begin{figure*}[t]
\centering
\includegraphics[width=\textwidth]{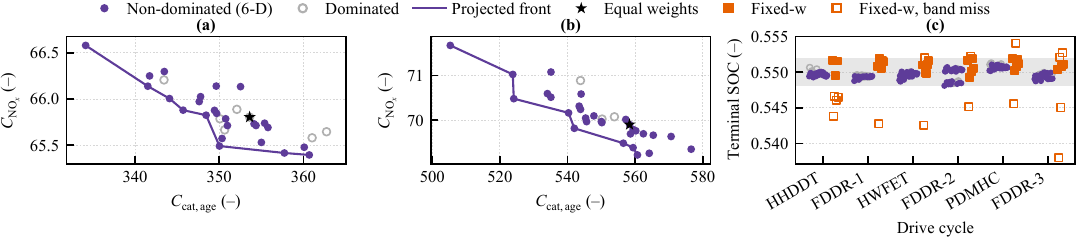}
\caption{Deterministic evaluation of the 32-preference grid at the
last iterate (representative run). (a), (b): catalyst aging against
tailpipe $\mathrm{NO}_x$ on HHDDT (in pool) and FDDR-3 (outside it);
filled circles are non-dominated in all six objectives, the line
connects those non-dominated in the projection, the star is the
equal-preference solution; fixed-weight solutions, extreme in the
objectives they ignore, are omitted. (c) Terminal SOC of all 32
solutions on all six cycles, with the seven fixed-weight trainings at
their own preference (squares, open if outside the band).}
\label{fig:fronts}
\end{figure*}

\subsection{Convergence and Terminal Feasibility}

The ten-run mean return improves from $-2{,}839$ to $-317$, with an
interquartile range of $45$ at the last update
(Fig.~\ref{fig:training}(a)). For the representative run the
equal-preference terminal SOC settles inside the band on all six
cycles over the last $240$ updates, with last-iterate landings between
$54.85\%$ and $55.06\%$ (Fig.~\ref{fig:training}(b)).

Whole-front feasibility is summarized in Fig.~\ref{fig:fronts}(c)
and Table~\ref{tab:results}. The representative run ends all $32$
solutions in band on every cycle. Across the ten runs the mean
is $31.6$--$31.9$ on FDDR-1, HWFET and FDDR-2, $28.5$ on PDMHC and
$26.7$ on FDDR-3, outside the training pool, where transfer varies
across runs. On HHDDT the mean is $23.2$; every solution of every run
lies in the band or no more than $0.2$ percentage points of SOC
outside it. The seven fixed-weight
trainings land three to six in band per cycle, at seven times the training cost
and with no intermediate trade-offs.

The PPO-L baseline reaches similar stochastic training feasibility, $99.9\%$ against $99.7\%$ over the last fifty
updates, and its front responds to $\w$. The difference is
deterministic transfer: on FDDR-3 its mean in-band count is $5.1$ of
$32$ against $26.7$, with median out-of-band miss $0.30\%$ SOC against
$0.19\%$. On HWFET, inside the pool, two of its ten runs place $26$
and $30$ of $32$ solutions $0.7$--$8.2$ percentage points of SOC below
the band despite stochastic feasibility above $99.7\%$ on that cycle. Among the in-band points, recomputed
six-dimensional non-dominance leaves $2.8$ baseline solutions against
$20.8$ for the proposed method on FDDR-3, and $13.6$ against $18.2$ on
HHDDT. The baseline also covers less of the objective space: its
spread over the preference grid is the smaller in five of the six
objectives on both HHDDT and FDDR-3, the widest gap being
$\mathrm{NH}_3$ slip at $8.1\%$ against $16.3\%$ and $9.2\%$ against
$16.3\%$ of the equal-preference value. Its equal-preference anchor is
in band on all six cycles in none of its ten runs, against six of ten,
and its return is lower, $-329.9$ against $-317.0$, so feasibility is
not exchanged for cost. This is the separation permitted by Theorem~\ref{thm:gap}:
comparable stochastic feasibility does not imply comparable
deterministic feasibility.

Feasibility is the stable quantity: every run places at least $28$ of
$32$ solutions in band on FDDR-1, HWFET and FDDR-2, whereas
the hypervolume of Table~\ref{tab:results}, over the feasible
non-dominated points of the ten runs normalized per cycle with
reference $1.1\,\bm{1}_m$, varies widely.

\begin{table}[t]
\centering
\caption{In-band solutions and hypervolume. Method entries:
mean\,/\,representative in-band counts out of $32$, and ten-run mean HV
with 25th--75th percentile range; PPO-L: ten-run means;
fixed-$\w$: counts of seven trainings.}
\label{tab:results}
\footnotesize
\setlength{\tabcolsep}{2.6pt}
\begin{tabular*}{\columnwidth}{@{}l@{\extracolsep{\fill}}cccc@{}}
\toprule
 & \multicolumn{2}{c}{Method (1 training, 32 solutions)} &
   PPO-L & Fixed-$\w$ \\
\cmidrule(lr){2-3} \cmidrule(lr){4-4} \cmidrule(l){5-5}
Cycle & in band & HV, mean [25--75\%] & in band & in band \\
\midrule
HHDDT  & $23.2\,/\,32$ & $0.16\;[0.04,0.30]$ & $15.9$ & $3/7$ \\
FDDR-1 & $31.7\,/\,32$ & $0.10\;[0.03,0.18]$ & $27.9$ & $6/7$ \\
HWFET  & $31.9\,/\,32$ & $0.09\;[0.03,0.15]$ & $25.0$ & $5/7$ \\
FDDR-2 & $31.6\,/\,32$ & $0.19\;[0.04,0.40]$ & $30.6$ & $5/7$ \\
PDMHC$^{\dagger}$  & $28.5\,/\,32$ & $0.10\;[0.03,0.20]$ & $28.1$ & $5/7$ \\
FDDR-3$^{\ddagger}$ & $26.7\,/\,32$ & $0.07\;[0.02,0.13]$ & $5.1$ & $3/7$ \\
\bottomrule
\multicolumn{5}{@{}l}{\scriptsize $^{\dagger}$monitoring only \quad
$^{\ddagger}$outside the training pool; neither is trained on}
\end{tabular*}
\end{table}

\subsection{Trade-Off Structure}

Fig.~\ref{fig:fronts}(a) and (b) project the strongest conflict,
catalyst aging against tailpipe $\mathrm{NO}_x$: the rank correlation
over the $32$ solutions is $-0.81$ on HHDDT and $-0.89$ on FDDR-3,
with $25$ and $29$ of them non-dominated in six dimensions. Lower
tailpipe $\mathrm{NO}_x$ requires higher dosing and thermal loading of
the catalyst, hence faster aging. Fuel is not a principal trade-off
dimension:
under the charge-sustaining constraint its cost varies by about $1\%$
across the front, against $8$--$13\%$ for catalyst aging and
$10$--$54\%$ for coverage tracking. Ordered by the first principal
component (Fig.~\ref{fig:heatmap}), the solutions move between an
emissions-favoring regime, in which tailpipe $\mathrm{NO}_x$,
$\mathrm{NH}_3$ slip and coverage error fall together while catalyst
aging rises, and its opposite; fuel is uncorrelated with the
ordering. For deployment, a solution is chosen from the feasible
non-dominated set and its preference is supplied to the conditioned
policy.

\begin{figure}[t]
\centering
\includegraphics[width=\columnwidth]{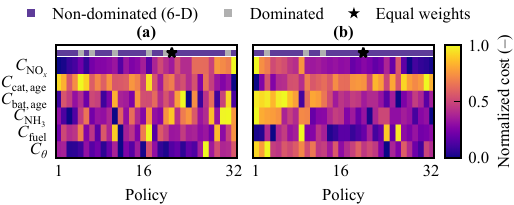}
\caption{Per-objective min--max normalized costs of the 32 solutions
(columns, ordered by the first principal component) on (a) HHDDT and
(b) FDDR-3; color encodes position within each objective. Top strip:
non-dominated (violet), dominated (grey), equal preference (star).}
\label{fig:heatmap}
\end{figure}

\subsection{Ablations}

The design was fixed by a one-factor study around a base that differs
from the method only in a relaxed ascent test,
$\bar{\xi}_c > 0.02\,\varepsilon$; each arm changes one element of the
base and retrains once. The exact test is
the only single change that brings the outside-pool anchor into the
band, at both floor levels ($54.92\%$ and $54.81\%$ against
$54.33$--$54.68\%$ for thresholds of $0.02$ to $0.10$): under a relaxed
test the multipliers rest on their floors ($\geq 97\%$ of the last
$100$ updates);
under the exact test the HHDDT multiplier stays above it. Moving the
normalization ahead of the preference weighting makes the validation
front land up to $0.3$ percentage points of SOC above the band and
vary little with the preference, consistent with the constraint-side
dilution of Proposition~\ref{prop:fidelity} and the displacement of
Corollary~\ref{cor:exclusion}; removing the floor as well does not
repair it, so the $2 \times 2$ check over placement and floor is
closed and, within the tested design, \eqref{eq:adv_sc} is necessary
rather than incidental.
Without the floor the multipliers decay by $1\%$ per update once the
stochastic feasibility test is passed, to $2$--$11\%$ of their initial
values, and the all-cycle in-band window shrinks to $20$ updates, the
mode excluded by Theorem~\ref{thm:floor}(ii); floors from $0.05$ to
$0.75$ of the cap vary the window between $0$ and $200$ updates and
leave all-cycle anchor feasibility at the last iterate unchanged. Removing the dense
terminal-band shaping
drives the multipliers to their caps, yet deterministic landings fall
up to $12$ percentage points of SOC below the band, with the anchor
outside it on every cycle, at stochastic training feasibility near
$88\%$: shaping teaches the end-of-trip recharge, the
floored multiplier calibrates its gain, and neither alone suffices.
The concentration $\alpha$, the anchor probability $p_{\mathrm{eq}}$,
the initial multiplier and the cap leave all-cycle anchor feasibility
at the last iterate unchanged.

\section{CONCLUSION}

A preference-conditioned extension of critic-free group-relative
policy optimization is presented for multi-objective control with a
terminal-state constraint. The analysis establishes three results:
under distinct normalization scales, per-objective normalization is
faithful only at the vertices of the simplex, a floored multiplier
prevents vanishing constraint pressure and renders the penalty exact
above a geometric multiplier, and a curvature bound separates
deterministic from stochastic feasibility. On heavy-duty
powertrain energy and emissions management with six stage costs, one
training run returns feasible trade-off fronts over a $32$-preference
grid on cycles inside and outside the training pool, at the cost of
one fixed-weight run. Within the tested design, normalization after
the weighting, the multiplier floor, the exact ascent test and the
terminal-band shaping are each required for deterministic feasibility
of the preference-indexed front. Future work will address full-front transfer outside the
pool and the coupled policy--multiplier iteration.

\section*{ACKNOWLEDGMENT}

Anthropic's Claude~\cite{c_anthropic2025} is used to assist with
manuscript editing and exposition. All technical content, algorithm
design, proofs and experimental results are the work of the authors.

\end{document}